\documentclass[12pt,a4paper,english,11pt]{article}
\usepackage[T1]{fontenc}
\usepackage[latin9]{inputenc}
\usepackage{amsmath}
\usepackage{amsthm}
\usepackage{amssymb}
\usepackage{graphicx}
\usepackage{setspace}
\makeatletter

\pdfpageheight\paperheight
\pdfpagewidth\paperwidth

\theoremstyle{plain}
\newtheorem{thm}{\protect\theoremname}
  \theoremstyle{definition}
  \newtheorem{defn}[thm]{\protect\definitionname}
  \theoremstyle{plain}
  \newtheorem*{cor*}{\protect\corollaryname}
  \theoremstyle{remark}
  \newtheorem{claim}[thm]{\protect\claimname}
  \theoremstyle{remark}
  \newtheorem*{rem*}{\protect\remarkname}
  \theoremstyle{plain}
  \newtheorem{lem}[thm]{\protect\lemmaname}

 \usepackage [ n,
advantage , operators , sets , adversary , landau , probability , notions , logic ,
ff, mm,primitives , events , complexity , asymptotics , keys]{ cryptocode}
\usepackage{hyperref}
\newcommand{\ket}[1]{|#1\rangle}
\newcommand{\bra}[1]{\langle #1|}
\DeclareMathOperator{\sval}{SVal}
\DeclareMathOperator{\pval}{PVal}
\DeclareMathOperator{\tval}{TVal}
\DeclareMathOperator{\tr}{Tr}
\usepackage{authblk}
\usepackage{fullpage}
\date{}

\makeatother

\usepackage{babel}
  \providecommand{\claimname}{Claim}
  \providecommand{\corollaryname}{Corollary}
  \providecommand{\definitionname}{Definition}
  \providecommand{\lemmaname}{Lemma}
  \providecommand{\remarkname}{Remark}
\providecommand{\theoremname}{Theorem}

\begin{document}
\author[1]{Maor Ganz} \author[1,2]{Or Sattath} \affil[1]{The Hebrew University} \affil[2]{MIT}

\title{Quantum coin hedging, and a counter measure}
\maketitle
\begin{abstract}
A quantum board game is a multi-round protocol between a single quantum
player against the quantum board. Molina and Watrous \cite{MW} discovered
quantum hedging. They gave an example for perfect quantum hedging:
a board game with winning probability $<1$, such that the player
can win with certainty at least $1\mbox{-out-of-}2$ quantum board
games played in parallel. Here we show that perfect quantum hedging
occurs in a cryptographic protocol \textendash{} quantum coin flipping.
For this reason, when cryptographic protocols are composed, hedging
may introduce serious challenges into their analysis. 

We also show that hedging cannot occur when playing two-outcome board
games in sequence. This is done by showing a formula for the value
of sequential two-outcome board games, which depends only on the optimal
value of a single board game; this formula applies in a more general
setting, in which hedging is only a special case. 
\end{abstract}

\section{Introduction}

\paragraph{Quantum board games}

A quantum board game is a special type of an interactive quantum protocol.
The protocol involves two parties: the player and the board. The board
implements the rules of the game: in each round $i$ of the protocol,
applies some quantum operation $O_{i}$, sends a quantum message to
the player, which can apply any operation it wants, and send a quantum
message back to the board. At the final round of the board game, the
board applies a two outcome measurement, which determines whether
the player won or lost. We assume that the player knows the rules
of the board game (the length of the messages, the operations $O_{i}$
and the two outcome measurement). The player has the freedom to decide
on his strategy \textendash{} the protocol does not specify what the
player should do in each round; the only constraint posed on the player
is that it must send a message of an appropriate length, as expected
by the board.

\paragraph{Perfect hedging}

Molina and Watrous showed that hedging is possible in quantum board
games~\cite{MW}. Prefect hedging is best explained by an example:
there exists a quantum board game for which no strategy can win with
certainty, but it is possible for a player to guarantee winning $1\mbox{-out-of-}2$
independent quantum board games, which are played in parallel. A formal
definition of hedging is given in Definition~\eqref{def:nHedging},
but for now, one can think of that example. In a follow up work, Arunachalam,
Molina and Russo~\cite{arunachalam2013quantum} analyzed a family
of quantum board games, and showed a necessary and sufficient condition
so that the player can win with certainty in at least $1\mbox{-out-of-}n$
board games. As discussed later, quantum hedging is known to be a
purely quantum phenomenon.

One example where Hedging becomes relevant is when reducing the error
(soundness) probability of quantum interactive proof protocols such
as QIP(2): since the optimal strategy for winning $t\mbox{-out-of-}n$
parallel repetitions is not necessarily an independent strategy, only
Markov bound (and not the Chernoff bound) can be used to show soundness
\cite{jain09two}. These aspects resembles the behavior that occurs
in the setting of Raz's (classical) parallel repetition theorem~\cite{raz98parallel};
the differences are that in the classical setting there are two players
who want to win all board games, whereas in our setting, there is
a single player, who wants to win at least $t\mbox{-out-of-}n$ board
games.

\paragraph{Coin flipping}

Quantum coin flipping is a two player cryptographic protocol which
simulates a balanced coin flip. When Alice and Bob are honest, they
both agree on the outcome, which is uniform on $\left\{ 0,1\right\} $.
Coin flipping comes in two flavors: Strong and weak. Perhaps the most
intuitive one is \emph{weak coin flipping}, in which each side has
an opposite desirable outcome: $0$ implies that Alice wins, and $1$
implies that Bob wins. An important parameter is the optimal winning
probability for a cheating player against an honest player. In weak
coin flipping we denote them by $P_{A}$ and $P_{B}$. We define $P^{*}=\max\left\{ P_{A},P_{B}\right\} $
\textendash{} the maximum cheating probability of both players. In
a \emph{strong coin flipping}, a cheating player might try to bias
the result to any outcome. We define $P_{A}^{0}$ to be the maximal
winning probability of a cheating Alice who tries to bias the result
to $0$, and $P_{A}^{1},P_{B}^{0},P_{B}^{1}$ are defined similarly.
In strong coin flipping $P^{*}=\max\left\{ P_{A}^{0},P_{A}^{1},P_{B}^{0},P_{B}^{1}\right\} $
that is $P^{*}$ bounds the possible bias to any of the outcomes,
by either a cheating Alice or a cheating Bob. In the classical settings,
it is known that without computational assumptions, in any coin flipping
protocol (either weak or strong) at least one of the players can guarantee
winning with probability $1$ ($P^{*}=1$) \cite{cleve86limits}.
Under mild computational assumption, coin flipping can be achieved
classically~\cite{blum83coin}. All of the results in the rest of
this paper hold information theoretically, that is, without any computational
assumptions. Unconditionally secure (i.e. without computational assumptions)
quantum strong coin flipping protocols with large but still non-trivial
$P^{*}<0.9143$ were first discovered by \cite{ATVY00}. Kitaev then
proved that in strong coin flipping, every protocol must satisfy $P_{0}^{*}\cdot P_{1}^{*}\ge\frac{1}{2}$,
hence $P^{*}\geq\frac{\sqrt{2}}{2}$ (\cite{Kit03}, see also \cite{Am}).
Therefore, the hope to find protocols with arbitrarily small cheating
probability moved to weak coin flipping. Protocols were found with
decreasing $P^{*}$(\cite{SR01,Amb04} showed strong coin flipping
with $P^{*}=\frac{3}{4}$, \cite{Moc04} showed weak coin flipping
with $P^{*}=0.692$), until it was finally proved that there are families
of weak coin flipping protocols for which $P^{*}$ converges to $\frac{1}{2}$~\cite{Moc}(see
also~\cite{ACGKM}). Following this, \cite{CK09} showed how such
protocol can be adopted, in order to create (arbitrarily close to)
optimal strong coin flipping (so that $P^{*}$ can be made arbitrarily
close to $\frac{\sqrt{2}}{2}$). Although this would not be relevant
for our work, analysis of coin flipping protocols was adapted, and
later implemented, for experimental setups \cite{pappa11practical,pappa14experimental}.
There is also a strong connection between coin-flipping and bit-commitment
protocols \cite{SR01,DBLP:conf/focs/ChaillouxK11}, and to a lesser
extent to oblivious transfer \cite{chailloux16optimal}.

Is it possible to hedge in quantum coin flips? In Section~\ref{sec:Quantum-coin-flip}
we give an example for perfect quantum hedging in the context of coin
flipping. The result can be best explained in the context of weak
coin flipping (although, a similar statement can be proved for strong
coin flipping): there exists a weak coin flipping protocol where $P^{*}=\cos^{2}(\frac{\pi}{8})$
introduced by Aharonov~\cite{aharonov07coin} yet a cheating Bob
can guarantee winning in at least $1\mbox{-out-of-}2$ board games
played in parallel.

\paragraph{Avoiding hedging through sequential repetition}

Consider a cryptographic quantum protocol, which involves several
uses of quantum two-outcome board games. For example, the protocol
may use several occurrences of quantum coin flips played in parallel.
As we have seen, the possibility of hedging makes it hard to analyze
the resulting protocol, by simply analyzing each of the board games
in it. In Section~\ref{sec:General-Hedging} we show that quantum
hedging cannot happen when the two-outcome board games are played
in sequence, even if the players are computationally unbounded. 

We give a more generalized formulation for sequential board games.
Suppose the player's utility for the outcome vector $a=(a_{1},\ldots,a_{n})$
is given by some target function $t(a)$, and the players goal is
to maximize $\mathbb{{E}}[t(a)]$ over all possible strategies. In
Theorem \ref{thm:SVAL=00003DTVAL} we show that this maximal value
is fully determined by the properties of each board game, and does
\emph{not} require an analysis of the entire system, which is the
case when playing in parallel. 

The authors are not aware of previous claims of that sort. The intuition
for the proof is fairly simple and arguably not very surprising: if
it is possible to hedge $n$ games, then by simulating the board in
the first game, and conditioning on some good event, allows the player
to hedge $n-1$ games. But since hedging cannot occur in one game,
we get a contradiction. 

Arunachalam, Molina and Russo \cite{arunachalam2013quantum} showed
a different approach to avoid hedging: they showed that hedging is
impossible in a quantum single round board game played in parallel,
where the player has the possibility to force a restart of the board
game. 

\section{\label{sec:Quantum-coin-flip}Quantum coin flip hedging}

In this section we will give an example for a coin flipping protocol,
for which a cheater cannot guarantee a win in one flip, but one of
the players can force a win in $1\mbox{-out-of-}2$ flips:
\begin{thm}
\label{thm:There-exists-a-WCF-hedging}There exists a weak coin flipping
protocol with $P^{*}<1$ s.t. by playing $2$ coin flips in parallel,
Bob can guarantee winning in at least one of the flips.
\end{thm}

We will first describe the weak coin flipping protocol and its properties,
and then analyze the hedging strategy of Bob. We conclude by explaining
why Alice cannot hedge.

\subsection{The coin flipping protocol}

In this work, Aharonov's coin flipping protocol~\cite{aharonov07coin}
will play an important role.\\
\\
\procedure{A quantum coin flipping protocol}{%
\textbf{Alice} \> \> \textbf{Bob} \\
\text{Prepares }\frac{1}{\sqrt{2}}(\ket{00}+\ket{11})\> \> \\
\> \sendmessageright*{\text{second qubit}} \> \\
\> \> \text{Samples } b\in_R\{0, 1\}.\\
\> \sendmessageleft*{\text{sends } b} \> \\
\text { If } b=1,\text{ then apply } H . \> \> \text { If } b=1,\text{ then apply } H .\\
\text{Measure in the standard basis}\> \> \text{Measure in the standard basis}\\
\text{Alice wins if the outcome is } 0  \> \>\text{\text{Alice wins if the outcome is } 0 }\\
\text{Bob wins if the outcome is } 1    \> \>\text{Bob wins if the outcome is } 1 \\
}
\begin{thm}
\label{thm:Protocol-1-achieves}The protocol above is a weak coin-flipping
protocol with  $P^{*}=P_{A}=P_{B}=\cos^{2}\frac{\pi}{8}$.
\end{thm}

The proof is given in Appendix~\ref{Proof of Theorem 1}.

\subsection{Coin hedging is possible}

Assume a cheating Bob plays two coin flips in parallel with an honest
Alice (it does not matter if he plays against the same person twice,
or against two different players, since they behave the same \textendash{}
because they are honest). We want to know the maximum probability
for a cheating Bob to win at least one coin flip. Surprisingly, this
is equal to $1$ in the protocol we previously described. This is
impossible if Bob were to play the two coin flips sequentially (see
Theorem~\ref{thm:winning1}).

We saw that for one coin flipping, $P_{A}=P_{B}=\cos^{2}\frac{\pi}{8}\approx0.853$.
By cheating each coin flip independently, the best Bob can get is
\[
\Pr\left(\text{Bob wins at-least one game}\right)=1-\left(1-P_{B}\right)^{2}=1-\left(1-\cos^{2}\frac{\pi}{8}\right)^{2}\approx0.978.
\]

We will now show Bob's perfect hedging strategy (which is not independent),
in which he wins exactly one out of the two coin flips w.p. $1$,
which completes the proof of Theorem~\ref{thm:Protocol-1-achieves}.
Alice's initial state is 
\begin{equation}
\frac{1}{2}\sum_{i_{1},i_{2}\in\left\{ 0,1\right\} }\ket{i_{1},i_{2}}\ket{i_{1},i_{2}}=\frac{1}{2}\sum_{i=0}^{3}\ket{\alpha_{i}}\ket{\alpha_{i}}\label{eq:AliceInit}
\end{equation}
, where\footnote{One may wonder whether the states $\ket{\alpha_{i}}$ are the Bell
states ($\ket{\Phi^{\pm}}=\frac{1}{\sqrt{2}}\left(\ket{00}\pm\ket{11}\right),\ \ket{\Psi^{\pm}}=\frac{1}{\sqrt{2}}\left(\ket{01}\pm\ket{10}\right)$),
written in a non-standard local basis. This is not the case: for every
Bell state $\ket{\Omega},$ $SWAP\ket{\Omega}=\pm\ket{\Omega}.$ This
is also true if a local basis change is applied to both qubits: for
$\ket{\Omega'}=U\otimes U\ket{\Omega}$, $SWAP\ket{\Omega'}=\pm\ket{\Omega'}$.
Since $\ket{\alpha_{2}}=SWAP\ket{\alpha_{3}}\neq\pm\ket{\alpha_{2}}$,
these vectors are not the Bell states written in a non-standard local
basis.}
\begin{align}
\ket{\alpha_{0}} & =\ket{\Phi^{-}}=\frac{1}{\sqrt{2}}\left(\ket{00}-\ket{11}\right)=\frac{1}{\sqrt{2}}\left(\ket{+-}-\ket{-+}\right)\nonumber \\
\ket{\alpha_{1}} & =\ket{\Psi^{+}}=\frac{1}{\sqrt{2}}\left(\ket{01}+\ket{10}\right)\nonumber \\
\ket{\alpha_{2}} & =\frac{1}{\sqrt{2}}\left(\ket{\Phi^{+}}-\ket{\Psi^{-}}\right)=\frac{1}{\sqrt{2}}\left(\ket{0-}+\ket{1+}\right)\nonumber \\
\ket{\alpha_{3}} & =\frac{1}{\sqrt{2}}\left(\ket{\Phi^{+}}+\ket{\Psi^{-}}\right)=\frac{1}{\sqrt{2}}\left(\ket{-0}+\ket{+1}\right).\label{eq:alpha_states}
\end{align}

Eq. \eqref{eq:AliceInit} can be justified by a direct calculation,
or by using the Choi\textendash Jamio\l kowski isomorphism \cite{choi75completely,jamiolkowski72linear},
see also \cite{watrous2011theory}, and noting that the associated
matrix for the l.h.s. and the r.h.s. are equal (both are proportional
to the identity matrix). Bob is given the right register of the state
above. Bob applies the unitary transformation $U=\sum_{i}\ket{\gamma_{i}}\bra{\alpha_{i}}$,
where $\ket{\gamma_{0}}=\ket{11},\ket{\gamma_{1}}=\ket{00},\ket{\gamma_{2}}=\ket{01},\ket{\gamma_{3}}=\ket{10}$,
so that the overall state becomes $\frac{1}{2}\sum_{i=0}^{3}\ket{\alpha_{i}}\ket{\gamma_{i}}$,
and sends the right register back to Alice. Alice measures the right
register in the standard basis (of course, Bob could have done this
just before sending the right register). The results of those measurements
determines the basis in which she measures the left register. This
strategy guarantees that Bob wins in exactly one coin flip: for example,
if Alice measures the qubits $|\gamma_{0}\rangle=\ket{11}$ then the
left register collapses to $|\alpha_{0}\rangle=\mid\Phi^{-}\rangle=\frac{1}{\sqrt{2}}\left(\mid+-\rangle+\mid-+\rangle\right)$,
and since in this case Alice measures both of the left register qubits
in the Hadamard basis, Bob will win in exactly one out of the two
coin flips. The right-most expressions in Eq. \eqref{eq:alpha_states}
are presented in this form so that it is easy to see the similar behavior
in the 3 other cases. 

One may wonder how strong the effect of hedging is. In particular,
can Bob guarantee $fn$ out of $n$ winnings, as long as $f\leq P^{*}$?
The answer is no: by playing three coin flipping of this protocol,
he cannot guarantee winning $2=\frac{2}{3}\cdot3$ with probability
$1$, even though $\frac{2}{3}\leq P^{*}$: we numerically calculated
that Bob can only win with probability $\approx0.986$ at least $2$
out of $3$ coin flips. This is still higher than the optimal independent
cheating that achieves a success probability of $\approx0.94$. 

Fortunately for Bob, Alice can not guarantee winning in $1\mbox{-out-of-}2$
parallel weak coin flipping. In fact, she cannot do any hedging. This
is true, essentially for the same reasons error reduction for QMA
works in a simple manner (vis-à-vis QIP(2)). The following argument
uses the definitions from Section~\ref{subsec:Perfect-hedging}.
Recall that from Bob's perspective, he is provided with a quantum
state given from Alice, and he measures it to determine whether he
wins or loses. Therefore $m(a_{i})=\min_{\ket{\psi_{i}}}\bra{\psi_{i}}M_{a_{i}}^{i}\ket{\psi_{i}}$,
which is equal to the smallest eigenvalue of $M_{a_{i}}^{i}$; and
$m^{par}(a_{1},\ldots,a_{n})=\min_{\ket{\psi}}\bra{\psi}M_{a_{1}}^{i}\otimes\cdots\otimes M_{a_{n}}^{i}\ket{\psi}$
which is equal to the smallest eigenvalue of $M_{a_{1}}^{i}\otimes\cdots\otimes M_{a_{n}}^{i}$.
But since $M_{a_{i}}^{i}$ is a measurement operator, its eigenvalues
are non-negative, and we conclude that $m^{par}(a_{1},\ldots,a_{n})=m(a_{1})\cdot\ldots\cdot m(a_{n})$.

\section{\label{sec:General-Hedging}How to circumvent hedging}

Our solution to circumvent hedging is to play the board games in sequence,
instead of in parallel. We will prove in Section \ref{subsec:Perfect-hedging}
that in the simple scenario, in which the goal is to win at least
$1\mbox{-out-of-}n$ sequential board games, hedging is not possible
(i.e. the best cheating strategy is to use the optimal cheating strategy
in each board game independently). We will generalize this in Section
\ref{subsec:General-hedging}, where we will prove that the same result
holds for every target function. Throughout this section, we will
consider only two-outcome board games (such as coin flipping), but
a generalization to any number of outcomes seems not too difficult
to achieve as well.

\subsection{\label{subsec:Perfect-hedging}Playing sequentially circumvents 1-out-of-n
hedging}

Molina and Watrous~\cite{MW} defined \emph{hedging} as the following
phenomenon.\footnote{Molina and Watrous restricted their definition to quantum board games
with a single round of communication (the board sends an initial quantum
state to the player, the player sends back another quantum state back
to the board, and then the board applies a measurement to determine
whether the player wins).} Suppose $G_{1},G_{2}$ are two board games with multiple outcomes
$A_{1},A_{2}$. For $a_{1}\in A_{1}$ let $m\left(a_{1}\right)$ be
the minimal probability that can be achieved for the outcome $a_{1}$
in $G_{1}$, and similarly for $m\left(a_{2}\right)$. If the board
game $G$ is not clear from the context, we may use $m^{G_{2}}(a_{2})$.
Now suppose that two board games are played in parallel, and the goal
is to minimize the probability for getting the outcome $a_{1}$ in
the first board game and $a_{2}$ in the second board game, which
is defined as $m^{par}\left(a_{1},a_{2}\right)$. Since the two strategies
can be played independently, clearly, $m^{par}\left(a_{1},a_{2}\right)\leq m\left(a_{1}\right)m\left(a_{2}\right)$.
\emph{Parallel Hedging} for two board games is the case where this
inequality is strict, that is $m^{par}\left(a_{1},a_{2}\right)<m\left(a_{1}\right)m\left(a_{2}\right)$.
Molina and Watrous gave an example for perfect parallel hedging in
which $m^{par}\left(a_{1},a_{2}\right)=0$ whereas $m\left(a_{1}\right)=m\left(a_{2}\right)>0$.
This definition can be naturally generalized to more than two board
games.
\begin{defn}
[Parallel Hedging]\label{def:nHedging}Let $G_{1},\ldots,G_{n}$
be $n$ quantum board games with possible outcomes $A_{1},\ldots,A_{n}$.
For $a_{i}\in A_{i}$, let $m\left(a_{i}\right)$ be the minimal probability
that can be achieved for the outcome $a_{i}$ in $G_{i}$. Similarly,
let $m^{par}\left(a_{1},\ldots,a_{n}\right)$ be the minimal probability
that can be achieved for outcomes $\left(a_{1},\ldots,a_{n}\right)$
when playing these $n$ board games in parallel. We say that hedging
is possible in $1\mbox{-out-of-}n$ board games if there exist $a_{1},\ldots,a_{n}$
s.t. 
\begin{equation}
m^{par}\left(a_{1},a_{2},\ldots,a_{n}\right)<\prod_{i=1}^{n}m\left(a_{i}\right).\label{eq:def_of_hedging}
\end{equation}
If $m^{par}\left(a_{1},a_{2},\ldots,a_{n}\right)=0$ and $\prod_{i=1}^{n}m\left(a_{i}\right)>0$,
then it is called \emph{prefect hedging}.
\end{defn}

It is known that inequality \eqref{eq:def_of_hedging} is actually
an equality in the classical case for single round board games \cite{MW,MS}.
We do not know whether the equality holds for multi-round classical
board games. What happens when the board games are played in sequence?
\begin{defn}
Given board games $\left\{ G_{i}\right\} _{i=1}^{n}$, the protocol
for playing the board games $\left\{ G_{i}\right\} $ in order is
called \emph{sequential}, assuming the player knows the result of
$G_{i}$ before the start of $G_{i+1}$ (this can be achieved by adding
a last round for each board game in which the board returns the outcome).
\end{defn}

Our next result shows that there is no sequential hedging for board
games (with any number of outcomes), and the cheater cannot do better
than to cheat each board game independently; that is if $\left\{ G_{i}\right\} _{i=1}^{n}$
are board games, then $m^{seq}\left(a_{1},\ldots,a_{n}\right)=m\left(a_{1}\right)\cdot\ldots\cdot m\left(a_{n}\right)$,
where $m^{seq}\left(a_{1},\ldots,a_{n}\right)$ is defined similarly
to $m^{par}\left(a_{1},\ldots,a_{n}\right)$ for sequential board
games. For simplicity and clarity, we will consider only the case
where all the board games are identical and $a_{i}=a_{j}=a$ for all
$i,j$, but the same proof will work for the general scenario as well
(one will just have to add indices indicating the board game for everything). 
\begin{thm}
\label{thm:winning1}Let $G$ be a board game, played sequentially
$n$ times, then $m^{seq}\left(a,\ldots,a\right)=m\left(a\right)\cdot\ldots\cdot m\left(a\right)$
for every outcome $a$.
\end{thm}

\begin{proof}
If the outcome of a single board game is $a$, then we say that the
player \emph{lost} that board game. We denote by ``\emph{failure}''
the event in which the player gets the outcome $a$ in all $n$ games
(i.e. loses all $n$ rounds).

We define $\ell^{*}$ to be the probability to get the outcome $a$
in the optimal strategy for one board game. Let $\ell_{n}$ be probability
to get the outcome $a$ over all the $n$-board games, in the best
independent strategy. It is easy to see that 
\begin{equation}
\ell_{n}=\min_{S\in\text{independent strategies}}\Pr\left(\text{failure}\mid S\right)=\left(\ell^{*}\right)^{n}\label{eq:l_n}
\end{equation}
 Define similarly $\ell{}_{n}^{\prime}$ to be the minimum loosing
probability over all (not necessarily independent) strategies, i.e.
$\ell_{n}^{\prime}\equiv\min_{S\in\text{sequential strategies}}\Pr\left(\text{failure}\mid S\right)$
. Clearly $\forall n\in\mathbb{N},\ \ell_{n}^{\prime}\leq\ell_{n}$
and $\ell_{1}^{\prime}=\ell_{1}$. Our goal is to show that $\forall n\in\mathbb{N},\ \ell_{n}^{\prime}=\ell_{n}$.
Assume towards a contradiction that this is not the case. Then there
exists a minimal $n>1$ for which $\ell_{n}^{\prime}<\ell_{n}$. 
\[
\left(\ell^{*}\right)^{n}\overset{\text{by }\eqref{eq:l_n}}{=}\ell_{n}>\ell_{n}^{\prime}=\ell_{n,L}^{\prime}\Pr\left(\text{lost first round}\right)\geq\ell_{n,L}^{\prime}\ell^{*}
\]
where $\ell_{n,L}^{\prime}:=\Pr\left(\text{failure}\mid\text{lost first round}\right)$.
The last inequality naturally holds because $\Pr\left(\text{lost first round}\right)\geq\ell^{*}$,
otherwise there exists a better strategy. Therefore, 

\[
\left(\ell^{*}\right)^{n-1}=\ell_{n-1}>\ell_{n,L}^{\prime}
\]
The strategy in which the cheater Alice plays with Rob (Alice's imaginary
friend) the first board game, and conditioned on losing, plays with
Bob the next rounds, has a losing probability $\ell_{n,L}^{\prime}$.

Therefore 
\[
\ell_{n-1}>\ell_{n,L}^{\prime}\geq\ell_{n-1}^{\prime}
\]
which contradicts the minimality of $n$.
\end{proof}
\begin{cor*}
Suppose the goal of a player is to win at least $1\mbox{-out-of-}n$
board games played sequentially. The optimal strategy is to play independently,
by using the optimal cheating strategy in each of the board games.
\end{cor*}

\subsection{\label{subsec:General-hedging}Playing sequentially circumvents any
form of hedging}

Let us consider a more general setting, in which the player's goal
is to maximize the expectation of some target function; i.e., for
a vector $t=(t_{a}\in\mathbb{R})_{a\in\left\{ 0,1\right\} ^{n}}$,
let
\[
\sval\left(t\right)\equiv\max_{S\in\text{sequential strategies}}\sum_{a\in\left\{ 0,1\right\} ^{n}}t_{a}\cdot\Pr\left(a\mid S\right)
\]
and similarly 
\[
\pval\left(t\right)\equiv\max_{S\in\text{parallel strategies}}\sum_{a\in\left\{ 0,1\right\} ^{n}}t_{a}\cdot\Pr\left(a\mid S\right).
\]

In general there are no relations between the parallel and sequential
values: in Appendix~\ref{sec:Parallel-Sequential-Relations} we give
a classical one round board game in which $\sval\left(t\right)>\pval\left(t\right)$
and another in which $\sval\left(t\right)<\pval\left(t\right)$ .
\begin{defn}
Given a two-outcome board game, let $q_{i}$ be the maximal probability
of the player to achieve the outcome $i\in\{0,1\}$.
\end{defn}

As we have seen before, the parallel value of a two-outcome board
game heavily depends on the details of the game. In contrast, the
sequential value is fully determined by $q_{0}$ and $q_{1}$.

In the following we will analyze the sequential value of the board
game. For that we will define the tree value function $\tval$, which
as the following theorem shows, is equal to the sequential value of
the board game. For simplicity we will assume that for all $i$, $G_{i}=G$,
but this can be easily extended for general $\left\{ G_{i}\right\} _{i=1}^{n}$.
\begin{defn}
For a vector $t=(t_{a})_{a\in\left\{ 0,1\right\} ^{n}}$ let $t_{b}^{\leftarrow}=t_{0b}$
and $t_{b}^{\rightarrow}=t_{1b}$. The tree value with parameters
$q_{0},q_{1}$ is defined as: 
\[
\tval\left(t\right)\equiv\max\left\{ q_{0}\tval\left(t^{\leftarrow}\right)+\left(1-q_{0}\right)\tval\left(t^{\rightarrow}\right)\right.,\left.q_{1}\tval\left(t^{\rightarrow}\right)+\left(1-q_{1}\right)\tval\left(t^{\leftarrow}\right)\right\} ,
\]
and for $c\in\mathbb{\mathbb{R}},$ $\tval(c)=c$.
\end{defn}

\begin{defn}
Consider a quantum board game $G$ played $n$ times in sequence.
A strategy is said to be \emph{pure black box} strategy if the strategy
used in the i-th board game is fully determined by the outcomes of
the previous board games. For a set $\mathcal{S}$ of strategies for
a single board game $G$, an $\mathcal{S}$-black-box strategy is
a pure black-box strategy in which the strategy at the i-th board
game (conditioning on previous outcomes) is in $\mathcal{S}$.
\end{defn}

\begin{thm}
\label{thm:SVAL=00003DTVAL}For every two-outcome board game (with
parameters $q_{0},q_{1}$), every $n$ and every $t\in\mathbb{R}^{2^{n}}$,
$\sval\left(t\right)=\tval\left(t\right)$.

Furthermore, its value can be obtained by an $\left\{ S_{0},S_{1}\right\} $-black-box
strategy, where $S_{0}$ ($S_{1}$) are any strategies that achieve
outcomes 0 (1) with probability $q_{0}$ ($q_{1}$).
\end{thm}

This theorem is in fact a generalization of Theorem \ref{thm:winning1}
for 2-outcome board games: By choosing $t_{a}=1-\delta_{a,a^{\prime}}$
we get that 
\begin{align}
\sval\left(t\right) & \equiv\max_{S\in\text{sequential strategies}}\sum_{a\in\left\{ 0,1\right\} ^{n}}t_{a}\cdot\Pr\left(a\mid S\right)=\max_{S\in\text{sequential strategies}}\sum_{a\neq a^{\prime}}\Pr\left(a\mid S\right)\nonumber \\
 & =\max_{S\in\text{sequential strategies}}1-\Pr\left(a^{\prime}\mid S\right)=1-\min_{S\in\text{sequential strategies}}\Pr\left(a^{\prime}\mid S\right)=1-m^{seq}\left(a^{\prime}\right).\label{eq:sval_and_m_seq}
\end{align}
By expanding the recursion, a simple inductive argument shows that
for our choice of $t$,
\begin{equation}
\tval(t)=1-m(a_{1})\cdot\ldots\cdot m(a_{n}).\label{eq:tval_m}
\end{equation}
By combining Theorem \ref{thm:SVAL=00003DTVAL} and Eqs. \eqref{eq:sval_and_m_seq}
and \eqref{eq:tval_m}, we reprove Theorem \ref{thm:winning1}.

\begin{figure}
\includegraphics[scale=0.75]{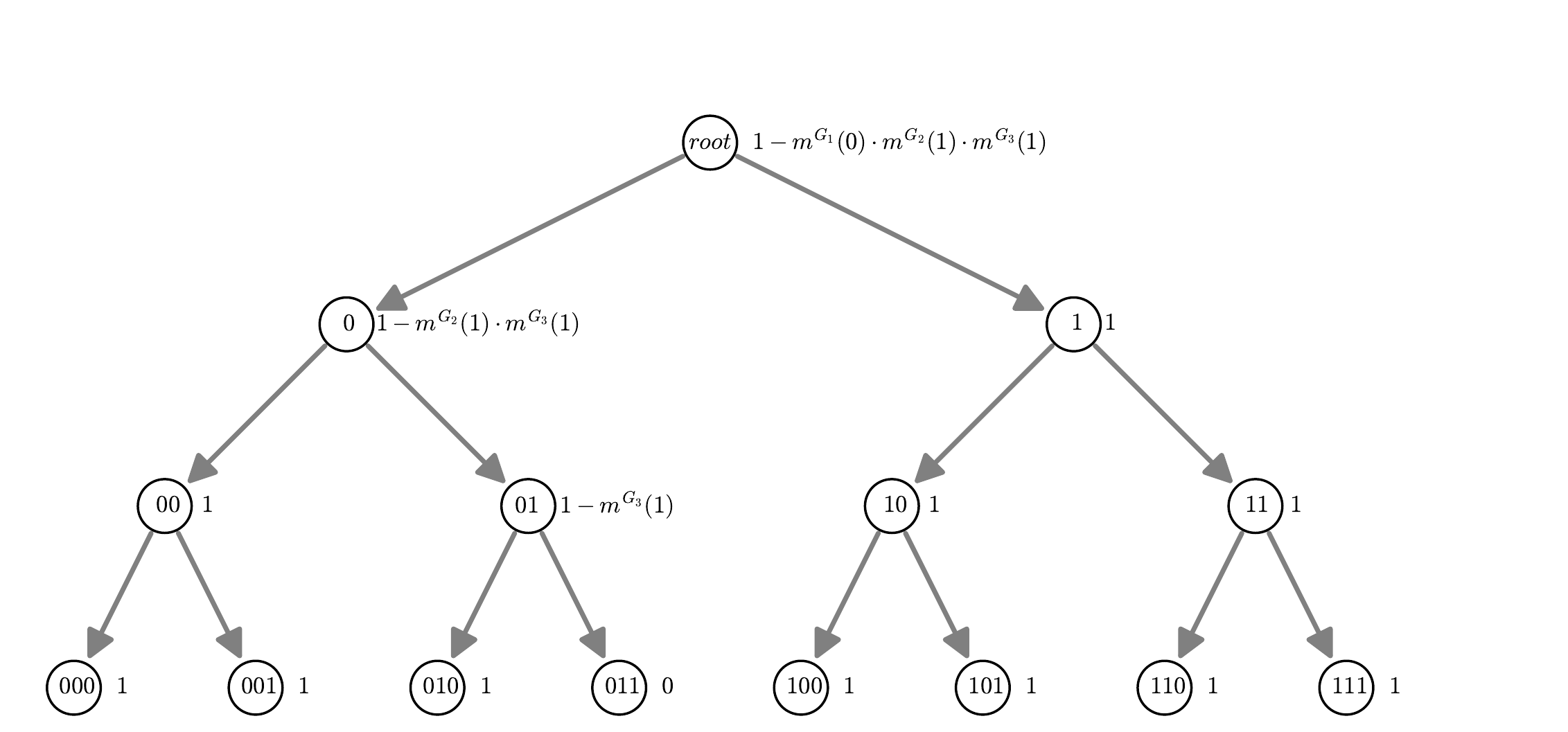}

\caption{$\tval$ for $t_{a}=1-\delta_{a,011}$. The labels of the leaves represent
all the possible outcomes $a$ of the values in the $n=3$ board games,
and the values on the right of each node are the $\tval$ of that
node. Indeed $t_{a}=1$ for all $a\protect\neq011$. Note that $m\left(0\right)=1-q_{1}$
and $m\left(1\right)=1-q_{0}$, and for example $\tval\left(01\right)=q_{0}=1-m^{G_{3}}\left(1\right)$,
and $\tval\left(0\right)=q_{0}+\left(1-q_{0}\right)q_{0}=1-m^{G_{2}}\left(1\right)+m^{G_{2}}\left(1\right)\left(1-m^{G_{3}}\left(1\right)\right)=1-m^{G_{2}}\left(1\right)\cdot m^{G_{3}}\left(1\right)$.}
\label{fig:tree}
\end{figure}
\begin{proof}
[Proof of Theorem~\ref{thm:SVAL=00003DTVAL}]First we show that $\sval\left(t\right)\geq\tval\left(t\right)$,
by explicitly constructing an $\{S_{0},S_{1}\}$-black-box strategy
with the value $\tval\left(t\right)$. The strategy can be best explained
by defining a binary full tree with depth $n$. We fill the value
of each node in the tree, from bottom to top. The leaves of the tree
will have values $t_{a}$. The values of a parent of two children
with values $v^{\leftarrow},v^{\rightarrow}$ will have the value:
\[
\max\{q_{0}v^{\leftarrow}+\left(1-q_{0}\right)v^{\rightarrow},q_{1}v^{\rightarrow}+\left(1-q_{1}\right)v^{\leftarrow}\}
\]
It can be easily verified that the value of the root is $\tval(t)$.

Consider the following strategy which applies $S_{0}$ if $q_{0}v^{\leftarrow}+\left(1-q_{0}\right)v^{\rightarrow}\geq q_{1}v^{\rightarrow}+\left(1-q_{1}\right)v^{\leftarrow}$
and $S_{1}$ otherwise, and continues in the same fashion with respect
to the left child if the outcome is $0$, and the right child if the
outcome is $1$. It can be proved by a simple inductive argument that
the expected value of this strategy is the value of the root which
is indeed $\tval(t)$. Clearly, this strategy is an $\left\{ S_{0},S_{1}\right\} $
black-box strategy.

Next we show that $\sval\left(t\right)\leq\tval\left(t\right)$. This
will be proven by induction on $n$ \textendash{} the number of board
games played. Clearly, for $n=1$, the optimal strategy has the value
$\tval(t)$. Let $n$ be the minimal number, such that there exists
some target $t$, for which there is a strategy with value greater
than $\tval(t)$ and denote the contradicting strategy by $S$. We
now introduce some notation. Let $p^{j}=\Pr\left(j\text{ in first game}\mid\text{using strategy }S\right)$,
$p_{\mathbf{i}}^{j}=\Pr\left(\mathbf{i}\text{ in the last n-1 games}\mid j\text{ in the first game, using strategy }S\right)$.
Let $\mathcal{S}^{n}$ be the set of all strategies over $n$ sequential
board games. 
\[
\text{opt}=\max_{S^{\prime}\in\mathcal{S}^{n}}\sum_{\mathbf{i}\in2^{n}}t_{\mathbf{i}}\Pr\left(\mathbf{i}\mid\text{ using strategy }S^{\prime}\right)
\]
For $j\in\{0,1\}$, let $\text{opt}^{j}\equiv\max_{S^{\prime}\in\mathcal{S}^{n-1}}\sum_{\mathbf{i}\in2^{n-1}}t_{j,\mathbf{i}}\Pr\left(\mathbf{i}\mid\text{using strategy }S^{\prime}\right)$.
Since the optimization is over board games of length $n-1$, by the
induction hypothesis, $\text{opt}{}^{0}=\tval(t^{\leftarrow})$, and
similarly $\text{opt}^{1}=\tval(t^{\rightarrow})$. We know that 
\begin{equation}
\text{opt}>q_{0}\cdot\text{opt}^{0}+\left(1-q_{0}\right)\cdot\text{opt}^{1}\label{eq:opt>q0}
\end{equation}
and similarly 
\begin{equation}
\text{opt}>q_{1}\cdot\text{opt}^{1}+\left(1-q_{1}\right)\cdot\text{opt}^{0}\label{eq:opt>q1}
\end{equation}
otherwise, $\text{opt}=\tval(t)$. Assume WLOG that 
\[
q_{0}\cdot\text{opt}^{0}+\left(1-q_{0}\right)\cdot\text{opt}^{1}\geq q_{1}\cdot\text{opt}^{1}+\left(1-q_{1}\right)\cdot\text{opt}^{0}
\]
then we get that $\text{opt}^{0}\left(q_{0}-1+q_{1}\right)\geq\text{opt}^{1}\left(q_{1}-1+q_{0}\right)$
hence $\text{opt}^{0}\geq\text{opt}^{1}$ or $\left(q_{1}-1+q_{0}\right)\leq0$,
because $q_{0}\geq1-q_{1}$. Since $p^{j}\leq q_{j}$ we get that
$q_{0}+q_{1}\leq1$ implies $p^{0}=q_{0}$ and $p^{1}=q_{1}$. We
know that 
\[
\text{opt}=\sum_{\mathbf{i}\in2^{n-1}}t_{\mathbf{i}}^{\leftarrow}p^{0}p_{\mathbf{i}}^{0}+t_{\mathbf{i}}^{\rightarrow}p^{1}p_{\mathbf{i}}^{1}.
\]
Let us denote 
\[
v^{0}=\sum_{\mathbf{i}\in2^{n-1}}t_{\mathbf{i}}^{\leftarrow}p_{\mathbf{i}}^{0}\ ,\ v^{1}=\sum_{\mathbf{i}\in2^{n-1}}t_{\mathbf{i}}^{\rightarrow}p_{\mathbf{i}}^{1}
\]
hence $\text{opt}=p^{0}v^{0}+p^{1}v^{1}$ where $p^{j}\leq q_{j}$.
\begin{claim}
$v^{j}\leq\text{opt}^{j}$
\end{claim}

\begin{proof}
The cheater can play himself (his honest self), according to his strategy,
until he gets $j$ in the first board game and then continue to play
the rest ($n-1$) of the board games against the real honest player.
This is a valid strategy for $n-1$ board games with value $v^{j}$,
but since $\text{opt}^{j}$ is an optimal such strategy, we get that
$v^{j}\leq\text{opt}^{j}$.
\end{proof}
Using the above claim,

\begin{equation}
\text{opt}=p^{0}v^{0}+p^{1}v^{1}\leq p^{0}\text{opt}^{0}+p^{1}\text{opt}^{1}=p^{0}\text{opt}^{0}+\left(1-p^{0}\right)\text{opt}^{1}.\label{eq:opt<}
\end{equation}
By subtracting Eq. \ref{eq:opt<} from Eq. \ref{eq:opt>q0} we get
that 
\[
0>\text{opt}^{0}\left(q_{0}-p^{0}\right)+\text{opt}^{1}\left(1-q_{0}-1+p^{0}\right)=\left(\text{opt}^{0}-\text{opt}^{1}\right)\left(q_{0}-p^{0}\right)
\]
but either $\text{opt}^{0}\geq\text{opt}^{1}$, $q_{0}\geq p^{0}$
and we get $0>0$ and contradiction, or $p^{0}=q_{0}$ hence again
we get $0>0$ and contradiction. Altogether we now know that Eq. \eqref{eq:opt>q0}
is wrong, hence 
\begin{equation}
\text{opt}=q_{0}\cdot\text{opt}^{0}+\left(1-q_{0}\right)\cdot\text{opt}^{1}\label{eq:opt=00003Dq0}
\end{equation}
and by the hypothesis assumption we get that $\text{opt}=\tval\left(t\right)$.
\end{proof}

\section{Open questions}
\begin{itemize}
\item Is there a formal connection between the setting discussed in the
parallel repetition Theorem (as was discussed in the introduction)
and the setting that occurs in quantum hedging?
\item How general is coin hedging? Does hedging (as in Definition \ref{def:nHedging})
happen in every non-trivial ($\epsilon<\frac{1}{2}$) coin flipping
protocol? The same questions can be asked for perfect hedging. We
conjecture that the answer for these questions is positive.
\item In our example for coin hedging, we saw that the hedging player reduces
the expected number of wins: The cheater could guarantee that he will
win one flip out of two, thus getting an expectation $0.5$ for winning,
while the expectation of winning in independent cheating is $\approx0.85$.
Does the expected ratio of wins in the perfect hedging of this protocol
scenario increase with $n$? In this protocol (or, perhaps, another
coin flipping protocol), when flipping $n$ coins in parallel and
$n\rightarrow\infty$, can Bob guarantee winning $\sim nP^{*}$ coin
flipping out of $n$?\\
This property cannot hold for every protocol. The reason is essentially
that $P^{*}$ can be artificially increased in a way which does not
help the cheating player to achieve perfect hedging. Consider some
coin flipping protocol with $P^{*}=\frac{1}{2}$ (even though this
is impossible, for $P^{*}>\frac{1}{2}$ a simple adaptation of the
following argument applies), then a cheating Bob clearly cannot guarantee
winning more than $\frac{1}{2}n$. If we now alter the protocol, such
that in the last round of the protocol, with probability $\delta$,
Alice asks Bob what his outcome of the protocol was, and declares
that as her outcome. This changes $P^{*}$ to $P^{*\prime}=\frac{1}{2}+\delta$,
but with probability $\delta^{n}$ these protocols coincide, and Bob
cannot guarantee more than $\frac{1}{2}n$ wins, which is less than
$P^{*\prime}n$ as required by the statement above.
\item Can one define and show hedging for bit-commitment?
\end{itemize}

\section{Acknowledgments}

We thank Dorit Aharonov for the weak coin flipping protocol which
we used and other valuable discussions, and to the anonymous referees
to their comments. This work was supported by ERC Grant 030-8301.

\bibliographystyle{alphaabbrurldoieprint}
\bibliography{hedging}


\appendix

\section{\label{Proof of Theorem 1}Proof of Theorem \ref{thm:Protocol-1-achieves}}
\begin{rem*}
The protocol is not only a weak coin flipping with $P^{*}=\cos^{2}\frac{\pi}{8}$,
but also a strong coin flipping protocol with the same value of $P^{*}$.
The proof is essentially the same. We state the result this way because
it provides a natural interpretation for statements such as ``Bob
wins in $1$ out of $2$ flips''. Of course, similar statements can
be made for strong coin flipping, but are omitted for the sake of
readability. 
\end{rem*}
We will use the same method, which is based on semi-definite programming
(SDP), we use in other sections. See, for example, \cite{Am}. We
will follow the notations used in \cite{ACGKM,Moc}. We will prove
that the maximal cheating probability for both players is $P^{*}=P_{A}=P_{B}=\cos^{2}\frac{\pi}{8}$.

If Alice is the cheater, a cheating strategy is described entirely
by the one qubit state $\rho$ which she sends to Bob. Her winning
probability is given by $\Pr\left(\text{Alice wins}\right)=\frac{1}{2}\tr\left(\left(\ket{0}\bra{0}+\ket{+}\bra{+}\right)\rho\right)$.
Since 
\begin{align*}
\max_{\rho\succeq0,\tr\rho=1}\frac{1}{2}\tr\left(\left(\ket{0}\bra{0}+\ket{+}\bra{+}\right)\rho\right) & =\max_{\ket{\psi}}\frac{\left\langle \psi\left|\frac{1}{2}\left(\ket{0}\bra{0}+\ket{+}\bra{+}\right)\right|\psi\right\rangle }{\left\langle \psi\mid\psi\right\rangle }\\
 & =\lambda_{max}\left(\frac{1}{2}\left(\ket{0}\bra{0}+\ket{+}\bra{+}\right)\right)\\
 & =\cos^{2}\frac{\pi}{8},
\end{align*}
the maximal cheating probability is $P_{A}=\cos^{2}\frac{\pi}{8}$. 

Let us look at a cheating Bob (and an honest Alice). The initial density
matrix is: $\rho_{0}^{\mathcal{AM}}=\left|\phi^{+}\left\rangle \right\langle \phi^{+}\right|$
on Alice and the message registers $\mathcal{A\otimes M}$. Then,
Bob applies an operation to the $\mathcal{M}$ qubit. Alice's reduced
density matrix cannot be changed due to Bob's operation. Hence our
condition is $\tr_{\mathcal{M}}\rho_{1}^{\mathcal{AM}}=\rho_{1}^{\mathcal{A}}=\rho_{0}^{\mathcal{A}}=\frac{1}{2}I$.
Bob's maximal cheating probability is given by:

\begin{alignat}{1}
\text{maximize}\quad & \tr\left[\left(\ket{1}\bra{1}\otimes\left|0\left\rangle \right\langle 0\right|+\left|-\left\rangle \right\langle -\right|\otimes\left|1\left\rangle \right\langle 1\right|\right)\cdot\rho_{1}^{\mathcal{AM}}\right]\label{eq:bob_win_1}\\
\text{subject to}\quad & \rho_{1}^{\mathcal{AM}}\succeq0\nonumber \\
 & \rho_{0}^{\mathcal{AM}}=\left|\Phi^{+}\left\rangle \right\langle \Phi^{+}\right|\nonumber \\
 & \tr_{\mathcal{M}}\rho_{1}^{\mathcal{AM}}=\rho_{0}^{\mathcal{A}}\nonumber 
\end{alignat}

The maximization is justified because if the message qubit is $0$,
Alice measures her qubit in the computational basis, and Bob wins
if her outcome is $1$; if the message qubit is $1$, Alice measures
her qubit in the Hadamard basis, and Bob wins if her outcome is $\mid-\rangle$.

Solving this SDP gives 
\[
\rho_{1}^{\mathcal{AM}}=\left(\begin{array}{cccc}
0.0732 & 0 & 0.1768 & 0\\
0 & 0.4268 & 0 & -0.1768\\
0.1768 & 0 & 0.4268 & 0\\
0 & -0.1768 & 0 & 0.0732
\end{array}\right)
\]
 with a maximum value of $\approx0.8536$. 

It is possible to verify that indeed the value of the SDP is not only
close, but is exactly equal to $\cos^{2}\frac{\pi}{8}\approx0.8536$:
One can see that $P_{B}\leq\cos^{2}\frac{\pi}{8}$, by finding an
explicit solution to the dual problem, or via Kitaev's formalism to
find the $Z$ matrix that bounds $\rho$ (see \cite{Moc,ACGKM} for
details). Alternatively, we can use the SDP formulation of games as
described in \cite{MW}, which applies to the coin-flipping protocol
(with Bob as the player): the matrix $Y=\frac{1}{8}\left(\begin{array}{cc}
3+\sqrt{2} & 1\\
1 & 1+\sqrt{2}
\end{array}\right)$ is dual-feasible, hence its trace $\tr\left[Y\right]=\frac{1}{4}\left(2+\sqrt{2}\right)=\cos^{2}\frac{\pi}{8}$
gives the correct bound.

We now show an explicit strategy with winning probability $\cos^{2}\frac{\pi}{8}$,
which shows that $P_{B}\geq\cos^{2}\frac{\pi}{8}$, which completes
the proof. Bob applies a $-\frac{3\pi}{8}$ rotation 
\[
U=\left(\begin{array}{cc}
\cos-\frac{3\pi}{8} & -\sin-\frac{3\pi}{8}\\
\sin-\frac{3\pi}{8} & \cos-\frac{3\pi}{8}
\end{array}\right)=\left(\begin{array}{cc}
\sin\frac{\pi}{8} & \cos\frac{\pi}{8}\\
-\cos\frac{\pi}{8} & \sin\frac{\pi}{8}
\end{array}\right)
\]
 on the $\mathcal{M}$ qubit, which transforms the state $\frac{1}{\sqrt{2}}\left(\ket{00}+\ket{11}\right)$
to:

\begin{align*}
\ket{\zeta} & =\frac{1}{\sqrt{2}}\left(\ket{0}\otimes\left(\sin\frac{\pi}{8}\ket{0}-\cos\frac{\pi}{8}\ket{1}\right)+\ket{1}\otimes\left(\sin\frac{\pi}{8}\ket{1}+\cos\frac{\pi}{8}\ket{0}\right)\right)\\
 & =\frac{1}{\sqrt{2}}\left(\left(\sin\frac{\pi}{8}\ket{0}+\cos\frac{\pi}{8}\ket{1}\right)\otimes\ket{0}\right)+\\
 & \frac{1}{\sqrt{2}}\left(\frac{1}{\sqrt{2}}\left(\left(\sin\frac{\pi}{8}-\cos\frac{\pi}{8}\right)\ket{+}-\left(\cos\frac{\pi}{8}+\sin\frac{\pi}{8}\right)\ket{-}\right)\otimes\ket{1}\right)
\end{align*}
We simplify
\[
\frac{1}{\sqrt{2}}\left(\sin\frac{\pi}{8}+\cos\frac{\pi}{8}\right)=\frac{1}{\sqrt{2}}\sqrt{\frac{1}{2}\left(2+\sqrt{2}\right)}=\frac{\sqrt{2+\sqrt{2}}}{2}=\cos\frac{\pi}{8}
\]
and similarly, $\frac{1}{\sqrt{2}}\left(\cos\frac{\pi}{8}-\sin\frac{\pi}{8}\right)=\frac{\sqrt{2-\sqrt{2}}}{2}=\sin\frac{\pi}{8}$.
Hence,
\[
\ket{\zeta}=\frac{1}{\sqrt{2}}\left(\left(\sin\frac{\pi}{8}\ket{0}+\cos\frac{\pi}{8}\ket{1}\right)\ket{0}-\left(\sin\frac{\pi}{8}\ket{+}+\cos\frac{\pi}{8}\ket{-}\right)\ket{1}\right).
\]
Bob measures the r.h.s. qubit in the computational basis, and sends
the classical result to Alice. His winning probability is thus $\cos^{2}\frac{\pi}{8}$.
This completes the proof that $P_{A}=P_{B}=P^{*}=\cos^{2}\frac{\pi}{8}$.

\section{\label{sec:Parallel-Sequential-Relations} Relations between parallel
and sequential board games}

Here we show that the value of the sequential board games can be larger
than the parallel board games and vice-versa, depending on the target
function, even in the classical setting. Our standard example for
a sequential superiority uses the target function: ``must win \emph{exactly}
$1\mbox{-out-of-}2$ board games''. This of course, gives the sequential
run an advantage over the parallel run, of knowing the outcomes of
the previous board games. For that we define a very simple one-round
board game: the player chooses a bit $b$, which is sent to the board. 
\begin{itemize}
\item If $b=0$, the player loses (with probability $1$).
\item If $b=1$, the player wins with probability $\frac{1}{2}$.
\end{itemize}
\begin{lem}
In the above board game, $\sval(t)\geq\frac{3}{4}>\frac{1}{2}=\pval(t)$.
\end{lem}

\begin{proof}
The optimal winning probability in a single board game for an honest
player is $\frac{1}{2}$ by always sending $b=1$. Also note, that
the player can force a loss with probability $1$, by sending $b=0$.
Assume that we are now playing two board games. If the board games
are played in sequence, then the optimal strategy will be to try and
win the first board game by sending $b_{1}=1$. With probability $\frac{1}{2}$
he will win, then he can lose the second board game by sending $b_{2}=0$.
If the player lost the first board game, he will try to win the second
board game by sending $b_{2}=1$. Altogether, this strategy wins exactly
once with probability $\frac{1}{2}+\frac{1}{4}=\frac{3}{4}$, proving
the first inequality.

Let us look at the four deterministic possibilities for the player
when the two board games are played in parallel. If he sends $b_{0}=b_{1}=0$,
he then loses with probability $1$. If he sends $b_{0}\neq b_{1}$
, i.e. loses one of the board games and tries to win the other, then
his winning probability of exactly one board game is $\frac{1}{2}$.
If he sends $b_{0}=b_{1}=1$, i.e. trying to win both, then his winning
probability of exactly one board game is again $\frac{1}{2}$ (because
no matter what the outcome of the first board game is, the second
outcome must be different, and this happens with probability $\frac{1}{2}$).
Since every random strategy is a convex combination of these deterministic
strategies, every classical strategy will also have a winning probability
of at most $\frac{1}{2}$, which is inferior to the winning probability
in the sequential setting. Naturally, giving the player quantum powers,
does not help him in this classical simple board game, to achieve
anything better.
\end{proof}
In the other direction, we give an example for a classical board game
in which the parallel setting, achieves better value than the sequential
one. Define a board game, in which the board sends a random bit $a$,
and then the player returns a bit $b$. If $a=0$, then the player
loses if $b=0$, and if $b=1$ then the player wins with probability
$p$. If $a=1$, then the player wins if $b=0$, and if $b=1$ then
the player loses with probability $p$. We think of $p$ to be of
a parameter  $p<\frac{3}{4}$. Our target function is the same as
before \textendash{} win \emph{exactly} $1\mbox{-out-of-}2$ board
games.
\begin{lem}
In the above board game, $\pval(t)\geq\frac{1}{2}+2p\left(1-p\right)>\sval(t)\geq\frac{3}{4}>\frac{1}{2}+\frac{1}{2}p=\sval(t)$.
\end{lem}

\begin{proof}
In the parallel settings, the player gets the $a_{1},a_{2}$ and only
then sends $b_{1},b_{2}$, which gives him the edge. If $a_{1}\neq a_{2}$,
his strategy is to send $b_{1}=0,b_{2}=0$ and he will win exactly
one board game out of the two. If $a_{1}=a_{2}$ then he will send
$b_{1}=b_{2}=1$ and he will win exactly one of the board games with
probability $p\left(1-p\right)$. Overall we see that $\pval(t)\geq\frac{1}{2}+2p\left(1-p\right)$.
In the sequential setting, it does not matter what happened in the
first board game, as the second board game will determine the result
(the outcome of the second board game must be different than the first).
With probability $\frac{1}{2}$ the board will send good $a_{2}$,
resulting in the player winning if he sends $b_{2}=0$ with certainty.
With probability $\frac{1}{2}$ the board will send bad $a_{2}$,
resulting in the player winning if he sends $b_{2}=1$ with probability
$p$. In total we get that $\sval(t)=\frac{1}{2}+\frac{1}{2}p$. By
taking $p<\frac{3}{4}$, we will get that $P_{seq}^{*}<P_{par}^{*}$. 
\end{proof}
In the quantum setting, we already saw that parallel can achieve better
value, in our coin flipping example in section \ref{sec:Quantum-coin-flip}.
We conclude that there is no general connection between the value
of the parallel setting and the sequential setting.

\end{document}